\documentclass[conference,9pt]{IEEEtran}

\usepackage{amsfonts,amsmath,amsthm,amssymb}
\usepackage{enumerate}
\usepackage{graphics,graphicx,subfigure}

\usepackage{color}
\usepackage{todonotes}
\usepackage{cancel}
\usepackage{url}
\usepackage{hyperref}
\usepackage{makeidx}
\usepackage{showidx}
\usepackage{multicol}        
\usepackage{xspace}
\usepackage{stmaryrd}       
\usepackage{pifont}          
\usepackage{fancybox}     
\usepackage{bm}

\newtheorem{theorem}{Theorem}
\newtheorem{proposition}{Proposition}
\newtheorem{corollary}{Corollary}

\newcommand{\vu}{{\bs u}}
\newcommand{\vv}{{\bs v}}

\newcommand{\vy}{{\bs y}}
\newcommand{\vz}{{\bs z}}
\newcommand{\vA}{{\bs A}}
\newcommand{\vB}{{\bs B}}
\newcommand{\vC}{{\bs C}}
\newcommand{\vD}{{\bs D}}
\newcommand{\vE}{{\bs E}}

\newcommand{\vI}{{\bs I}}

\newcommand{\vM}{{\bs M}}

\newcommand{\vX}{{\bs X}}
\newcommand{\vY}{{\bs Y}}
\newcommand{\vZ}{{\bs Z}}

\newcommand{\cA}{\mathcal{A}}

\newcommand{\cS}{\mathcal{S}}

\newcommand{\bP}{\mathbb{P}}

\newcommand{\bR}{\mathbb{R}}

\newcommand{\be}{\begin{equation}}
\newcommand{\ee}{\end{equation}}
\newcommand{\bal}{\begin{align}}
\newcommand{\eal}{\end{align}}
\newcommand{\ba}{\begin{align*}}
\newcommand{\ea}{\end{align*}}
\newcommand{\bmx}{\begin{matrix}}
\newcommand{\emx}{\end{matrix}}
\newcommand{\bbmx}{\begin{bmatrix}}
\newcommand{\ebmx}{\end{bmatrix}}
\newcommand{\bpmx}{\begin{pmatrix}}
\newcommand{\epmx}{\end{pmatrix}}
\newcommand{\bvmx}{\begin{vmatrix}}
\newcommand{\evmx}{\end{vmatrix}}

\newcommand{\ol}{\overline}

\newcommand{\wt}{\widetilde}
\newcommand{\f}{\frac}

\newcommand{\imp}{\Longrightarrow}

\newcommand{\inc}{\subseteq}

\newcommand{\lbt}{\llbracket}
\newcommand{\rbt}{\rrbracket}

\DeclareMathOperator{\sgn}{sgn}
\newcommand{\eps}{\varepsilon}
  
\newcommand{\range}[1]{\lbt #1 \rbt}
\newcommand{\bilr}[2]{{\Sigma}_{(#1)}^{[#2]}}
\newcommand{\ubilr}[2]{\wt{\Sigma}_{(#1)}^{[#2]}}
\newcommand{\bilrset}[2]{{\Sigma}_{#1}^{[#2]}}
\newcommand{\ubilrset}[2]{\wt{\Sigma}_{#1}^{[#2]}}
\DeclareMathOperator{\CW}{CW} % consistency width
\newcommand{\Amap}{\bs{\cl A}} % the linear operator, the bold caligraphic A
\newcommand{\bAmap}{\sgn \bs{\cl A}} % the binary operator

\DeclareMathOperator{\rank}{rank} % support
\newcommand{\bs}{\boldsymbol}
\newcommand{\bb}{\mathbb}
\newcommand{\cl}{\mathcal}

\newcommand{\ts}{\textstyle}

\newcommand{\ie}{\emph{i.e.},\xspace}
\newcommand{\eg}{\emph{e.g.}\xspace}
\newcommand{\iid}{%
    \ifmmode% math mode
        \mathrm{i.i.d.}%
    \else%
        i.i.d.\@\xspace%
    \fi%
}

\title{\vspace{-4mm}One-Bit Sensing of Low-Rank and Bisparse Matrices\vspace{-3mm}}
\author{\IEEEauthorblockN{Simon Foucart\IEEEauthorrefmark{1} and Laurent Jacques\IEEEauthorrefmark{2}}
\IEEEauthorblockA{\IEEEauthorrefmark{1}Texas A\&M University, USA. \IEEEauthorrefmark{2}UCLouvain, Belgium.}\vspace{-6mm}}

\begin{document}
\maketitle
\begin{abstract}
This note studies the worst-case recovery error of low-rank and bisparse matrices as a function of the number of one-bit measurements used to acquire them. First, by way of the concept of consistency width, precise estimates are given on how fast the recovery error can in theory decay. Next, an idealized recovery method is proved to reach the fourth-root of the optimal decay rate for Gaussian sensing schemes. 
This idealized method being impractical, 
an implementable recovery algorithm is finally proposed in the context of factorized Gaussian sensing schemes.
It is shown to provide a recovery error decaying as the sixth-root of the optimal rate.
\end{abstract}

%%%%%%%%%%%%%%%%%
%% The main text starts here %%
%%%%%%%%%%%%%%%%%

\section{Introduction}

We are interested in matrices that are simultaneously low-rank and bisparse,
\ie anticipating over some notation to be introduced momentarily, matrices in the set%
\begin{multline*}
\bilr{s}{r} 
:= \big\{
\vX  \in \bR^{n \times n}:
\rank(\vX) \le r,\, \exists S,T \inc \range{n},\\
|S|=|T|=s,\,\vX_{\ol{S} \times \range{n}} = {\bs 0},\,\vX_{\range{n} \times \ol{T}} = {\bs 0}\big\}.  
\end{multline*} 
Determining the minimal number of linear measurements allowing for uniform stable recovery of matrices in $\bilr{s}{r}$ is a difficult problem, \eg \cite{OJFEH} contains negative results for optimization-based recovery,
while \cite{FGR} showed, with some caveats, that $\Theta(rs\ln(n/s))$ measurement suffice. 
Here, we consider a similar problem for quantized linear measurements,
hoping for more favorable conclusions in this case.
Precisely, suppose that matrices $\vX \in \bilr{s}{r}$ are acquired via the one-bit measurements
\be
\label{Msts}
y_i = {\rm sgn}(\langle \vA_i , \vX \rangle_F )
= {\rm sgn} ({\rm tr} (\vA_i^\top \vX)),
\quad i \in \range{m}.
\ee
We ask: how large should $m$ be to make it possible to derive, from $\vy \in \{ \pm 1 \}^m$,
 an approximant $\vX'$ such that $\|\vX - \vX'\|_F \le \eps$,
 or, turning things around,
 how fast can $\|\vX - \vX'\|_F$ decay when $m$ increases?
 Phrased in terms of \emph{consistency width}, Section~\ref{SecCW}
 addresses this question with lower and upper estimates that almost match.
 A (weaker) upper estimate also appears in Section~\ref{SecIHT}, which introduces an idealized recovery scheme not enforcing consistency of the approximant. 
 This idealized scheme being impractical (probably NP-hard), 
 we exhibit in Section~\ref{Fapractalgo} 
 a multistep recovery scheme working with specific sensing schemes.
\\[-3mm]

{\bf Notation:} 
Given $\vM \in \bR^{n \times n}$, $\vM_{\Omega_1 \times \Omega_2}$ represents the submatrix of $\vM$ indexed by $\Omega_1 \times \Omega_2 \inc \range{n}\times \range{n}$, with $\range{n}:=\{1,\ldots,n\}$.
The set $\overline{\Omega} \inc \range{n}$ stands for $\range{n} \setminus \Omega$. 
The Frobenius norm is denoted by $\|\cdot\|_F$ and the Frobenius inner product by $\langle \cdot,\cdot \rangle_F$.
Given index sets $S,T \inc \range{n}$, the set $\bilrset{S \times T}{r}$ contains all rank-$r$ matrices supported on $S \times T$.
A tilde over a set $\cl K \inc \bR^{n \times n}$,
\eg $\ubilr{s}{r}$ or $\ubilrset{S \times T}{r}$,
means that this set is intersected with the sphere $\bb S_F^{n \times n}$ of $\bR^{n \times n}$ with respect to the Frobenius norm. 
The notation for the unit ball is $\bb B_F^{n \times n}$. 

We often rewrite \eqref{Msts} in more compact form by defining the map $\bs{\cl A}:\vM \in \bR^{n \times n} \mapsto \big( \langle \vA_i, \vM \rangle_F \big)_{i=1}^m \in \bb R^m$, so that 
\eqref{Msts} also reads $\vy = \bAmap(\vX)$ with $\bAmap := \sgn \circ\,\Amap$.
The sign operator $\sgn$ acts entrywise when applied to a vector.
We say that $\bs{\cl A}$ is a \emph{Gaussian random} map when the matrices $\bs A_1,\ldots,\bs A_m$ are independent $n\times n$ Gaussian random matrices, \ie $(\bs A_i)_{k,l} \sim_{\iid} \cl N(0,1)$.  
The adjoint of $\bs{\cl A}$ 
is the map $\bm{\cA}^*: \bs v \in \bR^m \to \sum_i v_i \bs A_i \in \bR^{n \times n}$.
All absolute constants are denoted by $C,C', c, c',\ldots$ and their values can change from line to line. 

\section{Consistency width}
\label{SecCW}

The $m^{\rm th}$ \emph{consistency width} of a set $\cl K \inc \bb S_F^{n \times n}$
is defined as
\begin{multline*}
\CW^m(\cl K) := \inf_{\bs{\cl A}: \bb R^{n \times n} \to \bb R^m}
\sup \{ \|\vX - \vX'\|_F: \vX,\vX' \in \cl K,\\
\; \bAmap(\vX) = \bAmap(\vX')\}.
\end{multline*}
Essentially, $\CW^m(\cl K)$ provides the best (with respect to $\Amap$) worst-case error of any \emph{consistent} recovery scheme, \ie one that outputs $\bs X' \in \cl K$ with the same one-bit measurements as the sensed matrix $\bs X \in \cl K$, see \cite{jacques2016error}.     
We shall show that
\be
\label{eq:two-bounds-for-CW}
c\,\frac{rs}{m}\ 
\le\ \CW^m\big(\ubilr{s}{r}\big)\ \le\ 
C\,\frac{rs}{m}\ \ln\left(\frac{nm}{rs}\right).
\ee

\subsection{Lower estimate}
Let us first establish the lower bound in \eqref{eq:two-bounds-for-CW}.
Given $\cl K \inc \bb S_F^{n \times n}$, we define the smallest
achievable worst-case error of any combination of one-bit maps $\sgn \bs{\cl A}$ with recovery schemes $\bs \Delta: \{\pm 1\}^m \to \cl K$ by
$$
\cl E^m(\cl K) := \inf_{(\bs{\cl A}, \bs \Delta)}
\sup_{\bs X \in \cl K} \|\bs X - \bs\Delta\big(\bAmap(\bs X)\big)\|_F.
$$ 
Note that $\bs \Delta$ is not required to output a consistent approximant.
However, the smallest achievable worst-case errors using consistent or unrestricted schemes are in fact equivalent, as emphasized below.
\begin{proposition}
For any $\cl K \inc \bb S_F^{n \times n}$, one has
\label{prop:CW-minimax-link}  
\begin{equation}
\label{eq:CW-minimax-link}    
\cl E^m(\cl K) \leq \CW^m(\cl K) \leq 2 \cl E^m(\cl K).
\end{equation}
\end{proposition}

\begin{proof} Given $(\bs{\cl A},\bs \Delta)$ and $\bs X, \bs X' \in \cl K$ sharing the same one-bit measurements, 
in view of $\bs \Delta(\bAmap(\bs X))= \bs \Delta(\bAmap(\bs X'))$, there holds
$\|\bs X 
- \bs X'\|_F 
\leq \|\bs X 
-\bs \Delta(\bAmap(\bs X)) \|_F 
+\|\bs X' 
-\bs \Delta(\bAmap(\bs X'))\|_F$, 
which yields the rightmost inequality in \eqref{eq:CW-minimax-link}.

Now, given $\bs{\cl A}$ and $\bs X \in \cl K$,
we choose $\bs \Delta(\bAmap(\bs X))$ as an arbitrary matrix $\bs X' \in \cl K$ consistent with $\vX$.
The leftmost inequality in  \eqref{eq:CW-minimax-link} follows easily.
\end{proof}

Proposition \ref{prop:CW-minimax-link}  implies in particular that the lower bound in~\eqref{eq:two-bounds-for-CW} can be derived from the following lower estimation of $\cl E^m(\ubilr{s}{r})$.

\begin{theorem}
\label{prop:lower-bound-minimax-bisparse-lr}
For an error level $\eps > 0$ , one has   
\begin{equation}
  \label{eq:lower-bound-result}
  \cl E^m\big(\ubilr{s}{r}\big) \leq \eps\ \imp m \geq C \eps^{-1} rs.  
\end{equation}
Equivalently, this means that $\cl E^m\big(\ubilr{s}{r}\big) \geq C\frac{rs}{m}$.
\end{theorem}

\begin{proof}
Suppose that $\cl E^m\big(\ubilr{s}{r}\big) \leq \eps$.
We fix subsets $S,T$ of $\range{n}$ with $|S| = |T| =  s$, a matrix $\bs M \in \bb R^{n \times r}$
such that $\bs M_{\overline{S}\times \range{r}} = \bs 0$, 
and consider the $rs$-dimensional subspace $\cl S$ of $\bb R^{n \times n}$ defined by 
$\cl S := \{\bs M \bs N: \bs N \in \bb R^{r \times n},\ \bs N_{\,\range{r} \times \overline{T}} = \bs 0\}$. 
Note that $\cl S \inc \bilr{s}{r}$,
hence $\cl E^m(\wt{\cl S}) \leq \cl E(\ubilr{s}{r}) \leq \eps$. 
Therefore, for any $\eps'>\eps$, say $\eps'= 2\eps$, we can find a pair $(\Amap, \bs \Delta)$ such that 
$\|\bs X - \bs\Delta\big(\bAmap(\bs X)\big)\|_F \le \eps'$ for all $\bs X \in \wt{\cl S}$.
This provides a $2\eps$-covering of $\wt{\cl S}$ with $|\bAmap(\cl S)|$ elements.
Recall, however, that the latter is much smaller than the obvious bound~$2^m$.
Indeed, \cite[Lemma 1]{JLBB} guarantees that
the cardinality of the image under $\bAmap$ of a $d$-dimensional subspace $\cl S'$ of $\bb R^{n \times n}$ satisfies $|\bAmap(\cl S')| \leq (\frac{2m}{d})^d$.
In our situation, we thus obtain a $2\eps$-covering of $\wt{\cl S}$ with log-cardinality at most $rs \ln(2m/(rs))$.
But, $\wt{\cl S}$ being isomorphic to the sphere $\bb S^{r \times s}_F$,
the log-cardinality of such a covering must be at least $ rs \ln(C/\eps)$.
Putting these two facts together gives $m \geq C \eps^{-1} rs$,
so that \eqref{eq:lower-bound-result} is proved.
\end{proof}

\subsection{Upper estimate}

Let us now establish the upper bound in \eqref{eq:two-bounds-for-CW}. 
We first show that, with high probability on a Gaussian random map $\Amap$, all pairs of matrices in a set $\cl K \inc \bb B^{n\times n}_F$ that are consistent under $\bAmap$ must be at most $\eps$ apart provided $m$ is large compared to $\eps$ and to the intrinsic set complexity. This complexity is measured by the {\em Kolmogorov entropy} of $\cl K$, \ie 
$$
\cl H(\cl K,\eta) := \ln\,\inf\{\,|\cl G|:~\cl G \inc \cl K \inc \cl G+\eta\bb B_F^{n\times n}\},
$$
where the addition is the Minkowski sum between sets. 
Thus, the smallest number of translated Frobenius balls with radius $\eta$ that cover~$\cl K$ is $\exp \cl H(\cl K,\eta)$. 
\begin{theorem}
\label{ThmCWUp}
Let $\bs{\cl A}: \bb R^{n \times n} \to \bb R^m$ be a Gaussian random map.
Given a subset $\cl K$ of $\bb S^{n \times n}_F$ and $\eps>0$, 
if
\begin{equation}
\label{eq:sample-complex-bound-consist-implic}
\ts m \ge \f{C}{\eps} \cl H\big(\cl K, c\,\eps\big),
\end{equation}
then it occurs with probability at least $ 1-\exp(-c \eps m)$ that,
for all $\vX,\vX' \in \cl K$,
\begin{equation}
\label{eq:consist-implic}
\bAmap(\vX) = \bAmap(\vX')\     
\imp \ 
\| \vX - \vX' \|_F \le \eps.
\end{equation}
\end{theorem}

This result is related, but not strictly equivalent, to Theorem~2.4 from \cite{oymak2015near}. 
Under a requirement similar to~\eqref{eq:sample-complex-bound-consist-implic}, it states that a Gaussian random map $\Amap$ defines, with high probability, a \emph{local} $\eps$-isometry of $\cl K$. In particular, it implies that, for all $\bs X, \bs X' \in \cl K$ such that $\|\bs X - \bs X'\|_F \geq \eps$, 
one has
\begin{equation}
\label{eq:local-isometry}    
\ts |\f{1}{m} \|\bAmap \bs X,\bAmap \bs X'\|_{H} - \f{1}{\pi} \arccos \langle \bs X, \bs X'\rangle| \geq c \eps,
\end{equation}
with $\|\bs s, \bs s'\|_{H}:=|\{i \in \range{m}: s_i \neq s'_i\}|$ for $\bs s,\bs s' \in \{\pm 1\}^m$. 
Note that \eqref{eq:local-isometry} does not yield the contraposition of \eqref{eq:consist-implic}, since it does not enforce that $\bAmap \bs X \neq \bAmap \bs X'$. 

\begin{proof}[Proof of Theorem~\ref{ThmCWUp}]
The stated version appears in \cite[Theorem 1.5]{BL15}.
For the sake of diversity of the arguments,
we follow the idea from \cite[Theorem~2]{JLBB}
and prove a result that would still imply the upper bound in \eqref{eq:two-bounds-for-CW},
but does require strengthening the assumption \eqref{eq:sample-complex-bound-consist-implic} to
\begin{equation}
\label{strengthened}
\ts m \ge \f{C}{\eps} \cl H\big(\cl K, c\,\f{\eps}{n}\big).
\end{equation}
We shall prove that $\bP_0 \le \exp(-c \eps m)$, where 
$$
\bP_0 := \small \bP\big(\exists \vX,\vX'\!\in \cl K: \|\vX - \vX'\|_F > \eps,\ 
\bAmap(\vX) = \bAmap(\vX')\big).
$$
For a number $\rho \in (0,\eps/2)$ yet to be chosen,
we consider a $\rho$-covering $\{\vX_1,\,\cdots,\vX_K\}$ of $\cl K$ with $K = \exp \cl H(\cl K, \rho)$. 
Given $\vX,\vX' \in \cl K$ satisfying 
$\|\vX - \vX'\|_F > \eps$ and $\bAmap(\vX) = \bAmap(\vX')$,
we select $k,k' \in \range{K}$ such that $\|\vX - \vX_k \|_F \le \rho$ and $\| \vX' - \vX_{k'} \|_F \le \rho$,
implying that
$\|\vX_k - \vX_{k'} \|_F \ge \eps - 2 \rho$
and that $\bAmap(\vX_k + \vZ) = \bAmap(\vX_{k'} + \vZ')$
where $\vZ := \vX - \vX_k$ and $\vZ' := \vX' - \vX_{K'}$ both have norm at most $ \rho$.
It follows that 
\begin{align*}
\bP_0 \le \bP\big( &\exists k,k' \in \range{K}: \|\vX_k - \vX_{k'}\|_F \ge \eps - 2 \rho,\\
& \exists \bs Z, \bs Z' \in \rho \bb B^{n \times n}_F:
\bAmap(\vX_k + \vZ) = \bAmap(\vX_{k'} + \vZ')\big).
\end{align*}
With $\Omega := \{ (k,k') \in \range{K}^2: \|\vX_k - \vX_{k'} \|_F \ge \eps - 2 \rho \}$, a union bound then gives
$$
\ts \bP_0 \leq \sum_{(k,k') \in \Omega} \bP_{k,k'},
$$
where the summands are
\begin{align*}
\bP_{k,k'} &\small := \bP \big(\exists \bs Z, \bs Z' \in \rho \bb B^{n \times n}_F:  \bAmap(\vX_k + \vZ) = \bAmap(\vX_{k'} + \vZ')\big)\\
&\small=\bP \big(\exists \bs Z, \bs Z' \in \rho \bb B^{n \times n}_F: \forall i \in \range{m},\\
&\small\hspace{9mm}
\sgn(\langle \vA_i, \vX_k + \vZ \rangle_F)
= \sgn(\langle \vA_i, \vX_{k'} + \vZ' \rangle_F) \big).
\end{align*}
Let us observe that, for $(k,k') \in \Omega$,
\begin{align*}
\bP_{k,k'} & \le \bP \big( \forall i \in \range{m}, \exists \bs Z, \bs Z' \in \rho \bb B^{n \times n}_F:\\
& \hspace{9mm}
\sgn(\langle \vA_i, \vX_k + \vZ \rangle_F)
= \sgn(\langle \vA_i, \vX_{k'} + \vZ' \rangle_F)  \big)\\
&\ts =\prod_{i=1}^m
\bP \big(\exists \bs Z, \bs Z' \in \rho \bb B^{n \times n}_F:\\
& \hspace{9mm}
\sgn(\langle \vA_i, \vX_k + \vZ \rangle_F)
= \sgn(\langle \vA_i, \vX_{k'} + \vZ' \rangle_F) \big) \\
&\ts \le \prod_{i=1}^m  \Big( 1 - \frac{1}{\pi} \arccos(\langle \vX_k, \vX_{k'} \rangle_F ) + \sqrt{\f{\pi}{2}} n \rho \Big),
\end{align*}
where the last step was an application of \cite[Lemma 9]{JLBB}.
In view of $\arccos(\langle \vX_k, \vX_{k'} \rangle_F ) \ge \|\vX_k - \vX_{k'}\|_F \ge \eps - 2 \rho$ and of the choice $\rho = c \eps / n$ for an appropriate $c$, we obtain
$$
\ts \bP_{k,k'} \le \prod_{i=1}^m \big( 1 - \f{\eps}{2 \pi} \big)
\le \exp( - \f{\eps}{2 \pi} m).
$$
Altogether, we conclude that
$$
\ts \bP_0 \le K^2 \exp(- \f{\eps}{2\pi} m) = \exp(2 \cl H(\cl K, c \f{\eps}{n}) - \f{\eps}{2\pi} m),
$$
which readily implies the result under the assumption \eqref{strengthened}.
\end{proof}

Interestingly, an explicit bound for the Kolmogorov entropy of the set of low-rank and bisparse matrices can be obtained as follows.

\begin{proposition}
  \label{prop:cover-number-bisparse-lowrank}
 Given $r,s\in \range{n}$ with $r \le s$ and $\eta > 0$, one has 
  \begin{align*}
  \cl H\big(\ubilr{s}{r}, \eta)&\leq 2s \ln(en/s) + r(2s + 1) \ln(9/\eta)\\
  &\leq C rs \ln \big(n/ (\eta s)\big).
  \end{align*}
\end{proposition}

\begin{proof}
Fixing subsets $S,T$ of $\range{n}$ with $|S| = |T| = s$, 
one can find
an $\eta$-covering $\cl G_{S,T}$ of
$\ubilrset{S \times T}{r}$ with cardinality $|\cl G_{S,T}| \leq (9/\eta)^{r(2s+1)}$, see \cite[Lemma 3.1]{CanPla}.
In view of $\ubilr{s}{r}\ =\ \bigcup_{S,T} \ubilrset{S}{r}$,
an $\eta$-covering of $\ubilr{s}{r}$ is then given by $\ts \cl G := \bigcup_{S,T } \cl G_{S,T}$. 
Using Stirling's bound, we derive that $\cl H(\ubilr{s}{r}, \eta) \leq \ln |\cl G| \leq \ln\big[ \binom{n}{s}^2 (9/\eta)^{r(2s +1)} \big] \leq 2s \ln(en/s) + r(2s + 1) \ln(9/\eta)$, which concludes the proof.
\end{proof}

 Theorem~\ref{ThmCWUp} and Proposition~\ref{prop:cover-number-bisparse-lowrank}
 combine into the following result. 
\begin{corollary}
\label{cor:consist-implic-bilr}
Let $\bs{\cl A}: \bb R^{n \times n} \to \bb R^m$ be a Gaussian random map.
Given $\eps>0$, if
\begin{equation}
\label{eq:cond-m-consist-bilr}
\ts m \ge \f{C}{\eps} rs \ln (\f{n}{\eps s}),
\end{equation}
then it occurs with probability at least $ 1-\exp(-c \eps m)$ that \eqref{eq:consist-implic} is fulfilled 
for all $\vX,\vX' \in \cl K$.
\end{corollary}

With Corollary \ref{cor:consist-implic-bilr} at hand, we can justify the upper bound in \eqref{eq:two-bounds-for-CW}.
For this purpose, consider $\eps > 0$ making \eqref{eq:cond-m-consist-bilr}
an equality.
The assumption of Theorem~\ref{ThmCWUp} then holds,  
so there is a linear map $\Amap: \bb R^{n \times n} \to \bb R^m$ such that
$$
\sup_{\substack{\vX,\vX' \in \ubilr{s}{r}\\\bAmap(\vX) = \bAmap(\vX')}}\ts
\|\vX - \vX'\|_F\ \le\ \eps.
$$
We conclude by remarking that 
$$
\ts \eps = \f{C}{m} rs \ln(\f{n}{\eps s})
\geq C \f{rs}{m}, \quad \mbox{and in turn} \quad
\eps \le \f{C}{m} rs\,\ln (\f{cnm}{rs}),
$$
which implies the desired upper bound on $\CW^m(\ubilr{s}{r})$.

\section{Idealized algorithm for Gaussian sensing}
\label{SecIHT}

Corollary~\ref{cor:consist-implic-bilr} --- as well as Theorem~\ref{ThmCWUp} for more general sets --- says that any $\vX \in \ubilr{s}{r}$ is well approximated by any other $\vX' \in \ubilr{s}{r}$ sharing the same one-bit measurements $\vy \in \{\pm 1\}^m$.
It does not, however, exhibit a scheme to create an approximant from $\vy$.
We now provide such a scheme
by considering the \emph{projected back projection}  
\be
\label{DefHT}
\vX' = \cl P_{\bilr{s}{r}}(\bm{\cA}^*\vy). 
\ee
Here, $\cl P_{\bilr{s}{r}}$ denotes the projection onto the set $\bilr{s}{r}$,
but we stress from the onset that computing it is not a realistic task, 
see \cite{FGR} for a related NP-hardness statement.
We note also that $\vX'$ need not be consistent, 
in the sense that $\bAmap(\vX')$ is not necessarily equal to $\vy$,
but that it has the same structure as $\vX$.
The essential tool behind our argument is a restricted isometry property for the set $\bilr{s}{r}$,
established below for Gaussian measurements.

\begin{theorem}
\label{ThmStandardRIP}
Let $\Amap: \bb R^{n \times n} \to \bb R^m$ be a Gaussian random map.
Given $\delta \in (0,1)$,
if
\be
\label{AssumGenRIP}
\ts m \ge \f{C}{\delta^2}  r s \ln ( \f{n}{s} ),
\ee
then it occurs with probability at least $ 1 - 2 \exp(- c \delta^2 m)$ that,
for all $\vZ \in \bilr{s}{r}$,
\be
\label{GenRIP}
\ts (1-\delta) \|\vZ\|_F
\le \f{\sqrt{\pi /2}}{m}\|\bm{\cA}(\vZ)\|_1
\le (1+\delta) \|\vZ\|_F.
\ee
\end{theorem}

\begin{proof}
If we overlook the exact powers of $\delta$,
the classical proof consisting of a concentration inequality followed by a covering argument will go through ---
the key ingredient being the estimate of \cite[Lemma~3.1]{CanPla} for the covering number of the ball of $\bilrset{S \times T}{r}$ combined with a union bound over all index sets $S,T$ of size~$s$.

Alternatively, in terms of Gaussian width, we can invoke the result from \cite{Sch} (see also \cite{plan2014dimension}) which guarantees that \eqref{GenRIP} holds provided
\be
\label{CondSch} 
m \ge C \delta^{-2} w \big(\ubilr{s}{r} \big)^2.
\ee
Then, in view of $w(\cup_{k=1}^K \cl K_k) \leq \max_k w(\cl K_k) + 3 \sqrt{\ln K}$ for all subsets $\cl K_1, \ldots, \cl K_K$ of the sphere
(see \cite[Lemma 12(ii)]{{BFNPW-Dict}}) and of the fact that $\ubilr{s}{r} = \cup_{S,T} \ubilrset{S \times T}{r}$ with $w(\ubilrset{S \times T}{r})^2 \leq C rs$, we deduce 
$$
\ts w\big(\ubilr{s}{r} \big)\leq C \sqrt{rs} + 3\sqrt{2 s \ln{(\f{en}{s})}} \leq C \sqrt{rs \ln(\f{n}{s})}.
$$
It is clear that \eqref{CondSch} holds under the assumption \eqref{AssumGenRIP}.
\end{proof}

We now state and prove the main result of this section.
The argument is an adaptation of the technique presented in \cite[Section~8.4]{Flavors}.

\begin{theorem}
\label{thm:PBP-error}
Assuming that \eqref{GenRIP} holds on $\ubilr{2s}{2r}$,
any $\vX \in \ubilr{s}{r}$ acquired via $\vy = \bAmap(\vX)$
is approximated by the matrix $\vX'$ defined in \eqref{DefHT}
 with error
\begin{equation}
\label{eq:DefHTbound}
\|\vX - \vX'\|_F \le C \sqrt{\delta}.
\end{equation}
\end{theorem}

\begin{proof} 
We write the singular value decompositions of $\vX$ and $\vX'$ as
$\ts\vX = \sum_{k=1}^r \sigma_k \vu_k \vv_k^\top$ and $\vX' = \sum_{k=1}^{r} \sigma_k \vu'_k \vv_k'^\top$,
with the vectors $\vu_k,\vv_k$ being supported on some $S,T$, $|S|=|T|= s$,
and the vectors $\vu'_k, \vv'_k$ being supported on some $S',T'$, $|S'|=|T'|=s$.
Let also $\cl P_\cS$ denote the orthogonal projection onto the linear space
$\cS = {\rm span}\left\{ \vu_1 \vv_1^\top,\ldots, \vu_r \vv_r^\top, \vu'_1 \vv_1'^\top,\ldots, \vu'_r \vv_r'^\top \right\} \inc \bilr{2s}{2r}$.
Notice that $\cl P_\cS(\vX) = \vX$ and $\cl P_\cS(\vX-\vX')=\vX-\vX'$ since $\vX \in \cl S$ and $\vX-\vX' \in \cl S$. 
Because the matrix $\vX'$ is the best approximant to $\bs \cA^* \bAmap(\vX)$ from $\bilr{s}{r}$,
it is a better approximant to $\bs \cA^* \bAmap(\vX)$ than $\vX$ is, \ie
$\| \bs \cA^* \bAmap(\vX)  - \vX' \|_F^2
\le \|\bs \cA^* \bAmap(\vX) - \vX \|_F^2$.
Introducing $\vX$ in the left-hand side and expanding the square yields
\begin{align*}
\|\vX - \vX'\|_F^2 & \le 2 \langle \vX - \vX', \vX - \bs \cA^* \bAmap(\vX) \rangle_F
\\
& = 2 \langle \cl P_\cS(\vX - \vX'), \vX - \bs \cA^* \bAmap(\vX) \rangle_F\\
& = 2 \langle \vX - \vX', \vX - \cl P_\cS(\bs\cA^* \bAmap(\vX)) \rangle_F\\
& \le 2  \|\vX - \vX'\|_F \| \vX - \cl P_\cS(\bs\cA^* \bAmap(\vX)) \|_F,
\end{align*}
which implies that
\be
\label{-1}
\|\vX - \vX'\|_F \le 2 \| \vX - \cl P_\cS(\bs\cA^* \bAmap(\vX)) \|_F .
\ee
Taking $\|\vX\|_F^2 = 1$ into account,
we remark that
\begin{align}
\nonumber
\|\vX - \cl P_\cS(\bs\cA^* \bAmap & (\vX)) \|_F^2 = 1+ \|\cl P_\cS(\bs\cA^* \bAmap(\vX))\|_F^2\\
\label{0}
&
 - 2 \langle \vX,  \cl P_\cS(\bs\cA^* \bAmap(\vX)) \rangle_F. 
\end{align}
Now, thanks to $\cl P_\cS(\bs\cA^* \bAmap(\vX)) \in \bilr{2s}{2r}$, we have
\begin{align*}
\nonumber\|\cl P_\cS(\bs\cA^*  \bAmap(\vX)) \|_F^2
&=\langle \bs\cA^* \bAmap(\vX), \cl P_\cS(\bs\cA^* \bAmap(\vX))  \rangle_F\\
\nonumber&=\langle \bAmap( \vX), \bs\cA( \cl P_\cS(\bs\cA^* \bAmap(\vX)) )  \rangle\\
\nonumber&\le \|\bs\cA(\cl P_\cS(\bs\cA^* \bAmap(\vX)))\|_1 \\
&\le (1+\delta)   \|\cl  P_\cS(\bs\cA^* \sgn(\cA \vX)) \|_F,
\end{align*}
which simplifies to
\be
\label{3}
\|\cl P_\cS(\bs\cA^* \sgn(\cA \vX)) \|_F \le 1+\delta .
\ee
Besides, thanks to $\cl P_\cS(\vX) = \vX$ and $\vX \in \bilr{s}{r}$, 
we also have
\begin{align}
\nonumber
\langle \vX,  \cl P_\cS(\bs\cA^* \sgn&(\cA \vX)) \rangle_F
 = \langle \vX,  \bs\cA^* \sgn(\cA \vX) \rangle_F\\
\label{2}
& = \langle \cA \vX,  \sgn(\cA \vX) \rangle
= \|\cA \vX\|_1 
\ge (1-\delta).
\end{align}
Putting \eqref{3} and \eqref{2} together in \eqref{0} gives
$$
\| \vX - \cl P_\cS(\bs\cA^* \bAmap(\vX)) \|_F^2
\le 1 + (1+\delta)^2 - 2(1-\delta) 
\le C \delta.    
$$
Finally, substituting in \eqref{-1}, we arrive at
$\|\vX - \vX' \|_F \le C \sqrt{\delta}$.
\end{proof}

To close this section,
we reformulate \eqref{eq:DefHTbound} by stating that, with high probability,
the recovery error  for projected  back projection with Gaussian random maps achieves the decay rate
$$
\|\vX - \vX'\|_F \le C \left( \f{r s \ln(n/s)}{m} \right)^{1/4},
$$
since the restricted isometry property \eqref{GenRIP} holds in the regime \eqref{AssumGenRIP}.

\section{Factorized Gaussian sensing}
\label{Fapractalgo}
While the projected back projection scheme \eqref{DefHT} is impractical,
changing the sensing map $\Amap: \bb R^{n \times n} \to \bb R^m$ 
allows for the design of a practical algorithm to recover matrices in $\ubilr{s}{r}$ with provable recovery error decaying as a root of  $rs \ln(n/s)/m$.  
The strategy follows a multistep argument put forward in \cite{iwen2017robust} for sparse phase retrieval
and also imitated in \cite{FGR} for low-rank and bisparse recovery from
standard, \ie not quantized, measurements.
It consists in defining a sensing map $\Amap: \bb R^{n \times n} \to \bb R^m$ associated with matrices 
$$
\bs A_i := \bs B^\top \bs A'_i \bs C \in \bb R^{n \times n},
\quad i \in \lbt m \rbt,
$$
where $\bs A'_1,\ldots, \bs A'_m \in \bb R^{p \times p}$
and $\bs B, \bs C \in \bb R^{p \times n}$ for some $p$ chosen later.
The one-bit measurements made on $\bs X \in \bb R^{n \times n}$
through the map $\Amap: \bb R^{n \times n} \to \bb R^m$  can also be interpreted as one-bit measurements made on $\bs B \bs X \bs C^\top \in \bb R^{p \times p}$ through the map $\Amap': \bb R^{p \times p} \to \bb R^m$ associated with the $\bs A'_i$.
Namely, in view of $\langle \bs B^\top \bs A'_i \bs C, \bs X \rangle_F = \langle \bs A'_i, \bs B \bs X \bs C^\top \rangle_F$,
we have 
\begin{equation}
\label{eq:fact-one-bit-sensing}
\bs y = \bAmap(\bs X) = \bAmap'(\bs B \bs X \bs C^\top).
\end{equation}
We will see that, under appropriate assumptions on $\Amap'$, $\bs B$, $\bs C$, 
a good approximant to $\bs X \in \ubilr{s}{r}$
constructed from $\bs y = \bAmap(\bs X)$ is
\begin{equation}
\label{eq:PBP-factorizedsensing}
\ts \bs X' := \cl H_{(s)}^{\rm col}\big\{ \cl H_{(s)}^{\rm row}\big[\,\bs B^\top \cl H^{[r]}(\Amap'^* \vy)\,\big]\,\bs C\big\}.    
\end{equation}
Here, $\cl H^{[r]}$ denotes the operator of best rank-$r$ approximation 
(obtained by keeping $r$ leading singular elements),
while $\cl H_{(s)}^{\rm row}$ (respectively $\cl H_{(s)}^{\rm col}$) stands for the operator of best approximation by $s$-row-sparse (respectively \mbox{$s$-column-sparse}) matrices, \ie obtained by keeping $s$ rows (respectively $s$ columns) with largest $\ell_2$-norms. The appropriate assumption on $\Amap'$ is a restricted isometry property of order $2r$ with constant $\delta' \in (0,1)$, reading
\begin{equation}
\label{AssumA'}
(1-\delta') \|\vZ \|_F
\le \|\Amap'(\vZ)\|_1
\le (1+\delta') \|\vZ\|_F
\end{equation}
whenever $\vZ \in \bR^{p \times p}$ has rank $\le 2r$.
The appropriate assumptions on $\vD = \vB, \vC$
are standard restricted isometry properties of order~$2s$ with constant $\delta \in (0,1)$, reading
\begin{equation}
\label{AssumBC}
(1-\delta) \|\vz\|_2^2
\le \|\vD \vz\|_2^2 \le
(1+\delta) \|\vz\|_2^2
\end{equation}
whenever $\vz \in \bR^n$ is $2s$-sparse.
It is routine to verify that the later condition guarantees that,
for any $k \ge 1$,
\begin{equation}
\label{eq:polar-rip-mtx}
\big| \langle (\vI - \vD^\top \vD) \vZ,\vZ' \rangle_F \big|
\le \delta \|\vZ\|_F \|\vZ'\|_F    
\end{equation}
whenever the row-supports of $\vZ,\vZ' \in \bR^{n \times k}$ have combined size at most $2s$.
Under this condition,
for any $s$-row-sparse $\vZ \in \bR^{n \times k}$ 
seen through the inexact measurements $\vY = \vD \vZ + \vE \in \bR^{p \times k}$, one has
\begin{equation}
\label{ObsHT}
\|\vZ - \cl H_{(s)}^{\rm row}[\vD^\top \vY] \|_F
\le 2 \delta \|\vZ\|_F + 2 \sqrt{2} \|\vE\|_F.
\end{equation}
Writing $\vZ' = \cl H_{(s)}^{\rm row}[\vD^\top \vY]$ and $\vM = \vZ - \vZ'$,
this comes from,
using in turn the defining property of $\vZ'$, \eqref{eq:polar-rip-mtx}, and \eqref{AssumBC},
\begin{align*}
&\ts \|\vM\|_F^2 = \|\vD^\top \vY\!-\!\vZ'\|_F^2 - \|\vD^\top \vY\!-\!\vZ\|_F^2 + 2 \langle \vZ\!-\!\vD^\top \vY, \vM \rangle_F\\
&\ts \leq 2 \langle \vZ \!-\! \vD^\top \vY, \vM \rangle_F
= 2 \langle  (\vI \!-\! \vD^\top \vD)\vZ, \vM  \rangle_F - 2 \langle \vE, \vD \vM \rangle_F\\
&\leq 2\delta \|\vZ\|_F \|\vM\|_F + 2 \|\vE\|_F
(1+\delta)^{1/2} \|\vM\|_F.
\end{align*}

We now state and prove the main result of this section.

\begin{theorem}
\label{prop:factsensing-RIP-bounds}
Given $\delta,\delta' \in (0,1)$, under assumptions \eqref{AssumA'} on $\Amap'$
and \eqref{AssumBC} on $\vB$ and $\vC$,
any $\bs X \in \ubilr{s}{r}$ acquired via $\vy = \bAmap(\bs X)$ is approximated by the matrix $\vX'$ given in \eqref{eq:PBP-factorizedsensing} with error 
\begin{equation*}
\|\bs X - \bs X'\|_F \leq C (\sqrt{\delta'} + \delta).    
\end{equation*}
\end{theorem}

\begin{proof}
In view of \eqref{eq:fact-one-bit-sensing},
the vector $\vy \in \{ \pm 1\}^m$ provides one-bit measurements made on the rank-$r$ matrix $\vB \vX \vC^\top \in \bR^{p \times p}$ through the sensing map $\Amap': \bR^{p \times p} \to \bR^m$.
Since this map satisfies the restricted isometry property \eqref{AssumA'},
it is known 
(see \cite[Theorem~2]{FouLyn} or even Theorem~\ref{thm:PBP-error} with $s=n=p$) that
\begin{equation}
\label{Multistep1} \| c_X \vB \vX \vC^\top - \cl H^{[r]}(\Amap'^*\vy)\|_F
\le C \sqrt{\delta'},
\end{equation}
with $c_X := \|\vB \vX \vC^\top\|^{-1}_F$. As an aside, we notice that both matrices $\vX \vC^\top \in \bR^{n \times p}$ and $\vX^\top \in \bR^{n \times n}$ are $s$-row-sparse, so that
\begin{eqnarray}
\label{FromRIPB}
(1-\delta) \|\vX \vC^\top \|_F^2 \le 
& \hspace{-2mm} \|\vB \vX \vC^\top \|_F^2 \hspace{-2mm}
& \le (1+\delta) \|\vX \vC^\top \|_F^2,\\
\label{FromRIPC}
(1-\delta) \|\vX^\top \|_F^2 \le
& \hspace{-2mm} \|\vC \vX^\top \|_F^2 \hspace{-2mm}
& \le (1+\delta) \|\vX^\top \|_F^2,
\end{eqnarray}
hence, since $\|\vX \vC^\top \|_F = \|\vC \vX^\top \|_F$ and $\|\vX^\top \|_F = \|\vX \|_F = 1$,
\begin{equation}
\label{Multistep2}
(1-\delta)^2 \le \|\vB \vX \vC^\top \|_F^2 \le (1+\delta)^2.
\end{equation}
Now, viewing $\cl H^{[r]}(\Amap'^*\vy)$ as inexact measurements made on the $s$-row-sparse matrix $\vZ := c_X \vX \vC^\top \in \bR^{n \times p}$ via $\vB$,
namely $\cl H^{[r]}(\Amap'^*\vy) = \vB \vZ + \vE$ with
$\vE := \cl H^{[r]}(\Amap'^*\vy) - c_X \vB \vX \vC^\top$,
\eqref{ObsHT} implies that 
\begin{multline}
\big\| c_X \vX \vC^\top - \cl H_{(s)}^{\rm row}\big[\,\bs B^\top \cl H^{[r]}(\Amap'^* \vy)\,\big]
\big\|_F\\
\le 2 \delta c_X \|\vX \vC^\top\|_F + 2 \sqrt{2} \|\vE\|_F
\label{Multistep3}
\le C\big(\delta + \sqrt{\delta'}\big),
\end{multline}
where the last inequality used \eqref{FromRIPB} and \eqref{Multistep1}.
Similarly, viewing $(\cl H_{(s)}^{\rm row}\big[\,\bs B^\top \cl H^{[r]}(\Amap'^* \vy)\,\big])^\top
= \vC \vZ' + \vE'$
 as inexact measurements made on the \mbox{$s$-row-sparse} matrix $\vZ' := c_X \vX^\top \in \bR^{n \times n}$ via $\vC$
with
$\vE' := (\cl H_{(s)}^{\rm row}\big[\,\bs B^\top \cl H^{[r]}(\Amap'^* \vy)\,\big])^\top - c_X \vC \vX^\top$,
\eqref{ObsHT} implies that 
\begin{multline}
\big\| c_X \vX^\top - \cl H_{(s)}^{\rm row}
\big\{\, \vC^\top (\cl H_{(s)}^{\rm row}\big[\,\bs B^\top \cl H^{[r]}(\Amap'^* \vy)\,\big])^\top \,\big\}
\big\|_F\\
\le 2 \delta c_X \|\vX^\top\|_F + 2\sqrt{2}\|\vE'\|_F
\label{Multistep4}
\le C\big(\delta + \sqrt{\delta'}\big),
\end{multline}
where the last inequality used \eqref{Multistep2} and \eqref{Multistep3}.
By noticing that $\cl H_{(s)}^{\rm row}(\vM^\top) = \big( \cl H_{(s)}^{\rm col}(\vM) \big)^\top$,
\eqref{Multistep4} actually reads
$$
\| c_X \vX - \vX'\|_F
\le C \big (\delta + \sqrt{\delta'}\big).
$$
The final result follows from a triangle inequality and the fact that 
$\| \vX - c_X \vX \|_F
= |\|\vB \vX \vC^\top\|_F - 1| / \|\vB \vX \vC^\top\|_F
\le C \delta$,
which is a consequence of \eqref{Multistep2}.
\end{proof}

We conclude by remarking that assumption \eqref{AssumA'} is valid when $\Amap': \bR^{p \times p} \to \bR^m$ is a properly normalized Gaussian random map with $m \asymp \delta'^{-2}rp$
(see \cite[Theorem~1]{FouLyn} as an instantiation of \cite{Sch})
and that assumption \eqref{AssumBC} is valid when $\vB,\vC \in \bR^{p \times n}$ are a properly normalized Gaussian random matrices with $p \asymp \delta^{-2} s \ln(n/s)$.
Choosing $\delta' = \delta^2$,
we therefore arrive, with high probability, at $\|\vX - \vX'\|_F \le C \delta$ with
$m \asymp \delta^{-6} r s \ln(n/s)$, \ie 
$$
\|\vX - \vX'\|_F \le C \left( \f{r s \ln(n/s)}{m} \right)^{1/6}.
$$
We note that the power $1/6$ is certainly not optimal ---
it can \eg be increased by replacing $\cl H^{[r]}$ by a scheme improving on $\sqrt{\delta'}$ in \eqref{Multistep1}.

\section{Acknowledgments}
S.F. is supported by NSF grants DMS-1622134 and DMS-1664803.  
L.J. is funded by the F.R.S.-FNRS.
L.J. thanks L.~Demanet for interesting discussions on factorized Gaussian sensing for sparse phase retrieval, \eg on \cite{iwen2017robust}.

\end{document}